\documentclass{article}
\usepackage[latin1]{inputenc}
\usepackage[T1]{fontenc}
\usepackage[usenames]{color}
\usepackage{graphicx}
\usepackage{amsfonts,amssymb,latexsym}
\usepackage{amsmath}
\usepackage{amsthm}
\newtheorem {definition} {Definition}
\newtheorem {Th} {Theorem}
\newtheorem {Prop} {Proposition}
\newtheorem {Lem} {Lemma}
\newtheorem {Cor} {Corollary}
\newtheorem {req} {Remark}

\newcommand\Z{{\mathbb Z}}

\newcommand\R{{\mathbb R}}
\newcommand\C{{\mathbb C}}

\begin{document}
\title{Dynamical Bounds  for Sturmian Schrödinger Operators}

\author{L. Marin}

\maketitle

\renewcommand{\abstractname}{Abstract}
\begin{abstract}

The Fibonacci Hamiltonian, that is a Schr\"{o}dinger operator associated to a quasiperiodical 
sturmian potential with respect to the golden mean has been
investigated intensively in recent years.
Damanik and Tcheremchantsev developed a method in  \cite{dam1} and find
a non trivial dynamical upper
bound for this model. In this paper, we use this method to  generalize to a large family of Sturmian operators dynamical upper bounds and show  at sufficently large coupling 
anomalous transport for  operators associated to irrational number with a generic diophantine condition. As a counter example, we
exhibit a pathological  irrational number which do not verify this condition and show its associated dynamic exponent only has ballistic bound. Moreover, we establish a global lower
bound for the lower box counting dimension of the spectrum that is used to
obtain a dynamical lower  bound for  bounded density irrational numbers.

\end{abstract}

%INTRODUCTION
\section{Introduction}
If $H$ is a self-adjoint operator on a separable Hilbert space $\mathcal{H}$, the time dependent Schr\"{o}dinger equation of quantum mechanics, $i \partial_t \psi = H \psi$, yields to a unitary dynamical evolution in $\mathcal{H}$,
$$
\psi (t) = e^{-itH}\psi (0).
$$

Under the time evolution, $\psi (t)$ will generally spread out with time. This could be a complicated question  to quantify this spreading in concrete cases. One of the most studied case is where $\mathcal{H}$  is given by $L^2(\mathbb{R}^d)$ or $l^2(\mathbb{Z}^d)$, $H$ is a Sch\"{o}dinger operator of the form $-\Delta +V$, and $\psi(0)$ is a localized wavepacket. The form of the potential $V$ is depending on the physical model one studies. One of the most studied is the Sturmian potential and its particular subcase, the Fibonacci Hamiltonian, describing a standard one-dimensional quasicrystal.

The first approach to study quantum dynamics is the spectral theorem. Recall that each initial vector $\psi (0) = \psi$ has a spectral measure, defined as the unique Borel measure verifying
$$
\langle \psi , f(H) \psi \rangle = \int_{\sigma(H)} f(E)d_{\mu_\psi} (E)
$$
for every measurable function $f$. $\langle.,.\rangle$ denotes the scalar product of $\mathcal{H}$. A major step in the theory discovered by Guarneri (\cite{Gua1}, \cite{Gua2}) was that suitable continuity properties of the spectral measure  $d_{\mu_\psi}$ implies lower bounds on the spreading of the wavepacket. It was then extended  by many authors in \cite{com}, \cite{Gua3}, \cite{Last}, \cite{kl}. Continuity properties of the spectral measure follows from upper bounds on measure of intervals, $\mu_\psi ([E-\varepsilon,E+\varepsilon]), E \in \sigma (H), \varepsilon \to 0$. Later on, many authors refined Guarneri's method (\cite{bgt}, \cite{Gua4}, \cite{tche1}) allowing to take into account the whole statistics of $\mu_\psi ([E-\varepsilon,E+\varepsilon]), E \in \mathbb{R}$. One can find better lower bounds with information about both measure of intervals and the growth of the generalized eigenfunctions $u_\psi (n,E) $ (\cite{kl}, \cite{tche1}).

In the case of Schr\"{o}dinger operators in one space dimension, the information on the spectral measure and on generalized eigenfunctions is linked to the properties of  solutions to the difference (also called sometimes free) equation $Hu=Eu$ (\cite{dam4}, \cite{dam2}, \cite{ger1}, \cite{jit1}, \cite{jit2}, \cite{tche4}).
 Explicit lower bounds on spreading rate for numerous concrete cases come from an analysis of these solutions (\cite{dam5}, \cite{dam4}, \cite{ger1}, \cite{jit1}, \cite{jit2}, \cite{kl}).

The second approach to dynamical lower bounds in one dimension is based on the Parseval formula,
$$
2\pi \int_0^\infty e^{-2t/T} |\langle e^{-itH} \delta_1,\delta_n \rangle |^2 dt = \int_{-\infty}^\infty |\langle (H -E -\frac{i}{T})^{-1} \delta_1,\delta_n \rangle |^2 dE.
$$
This method developed in \cite{dam6}, \cite{dam7} and \cite{tche4} is the basis for the results in \cite{dam8} and \cite{jit3}. This method has the advantage that it gives directly dynamical bounds without any knowledge of the properties of spectral measure.  What is required is  upper bounds for solutions corresponding to some set of energies, which can be very small (non empty is sufficient). Moreover, additional information allows to  improve the results. A combination of both approach leads to optimal dynamical bounds for growing sparse potentials (see \cite{tche4}).

As mentioned before, there is a fairly good  understanding of how to prove dynamical lower bounds, specially in one space dimension. Results of dynamical upper bounds are a few and more recent. Proving upper bounds is hard because one needs to control the entire wavepacket. In fact, the dynamical lower bounds that typically established only bound some (fast) part of the wavepacket from below and this is sufficient for the desired growth of the standard dynamical quantities. In the same way, it is of course much easier to prove upper bounds only for a (slow) portion of the wavepacket. Killip, Kiselev and Last developed this idea with success in \cite{kkl}. Their work provides explicit criteria for upper bounds on the slow part of the wavepacket in terms of lower bounds on solutions.  Applying their general method to the Fibonacci operator, their result supports the conjecture that this model exhibits anomalous transport (i.e. neither localized, nor diffusive, nor ballistic).
%There is a lot of numerical and heuristic evidence for this claim in the physics literature.%

The conjecture for Fibonacci model is finally proved at sufficiently large coupling by Damanik and Tcheremchantsev in \cite{dam1}. They developed a general method establishing a connection between solutions properties and dynamical upper bounds.
Based on the parseval formula, this method allows to bound the entire wavepacket from above provided that suitable lower bounds for solution (or rather transfer matrix) growth at complex energies are available.

It is the main purpose of this paper to extend the application of this general method used for concrete Fibonacci model  to almost every Sturmian potentials. We will show that one has anomalous transport for Sturmian model associated to irrational number far enough from rational numbers, in a sense we develop further. On the contrary, we construct an irrational number close enough to rational number that yields to balistic motion.

In this paper, we use tools that are relevant to give a new lower  bound for the box dimension of the spectrum that is better for almost every irrational numbers. Since the spectrum is a Cantor set with Lebesgue measure zero, it is logical to investigate its fractal dimension.
It is well known that this  Cantor set is the limit of band spectra of approximant operators \cite{s1,bel}.
To find the bound, we use band spectra at rank $n$ as a sequence of $\varepsilon_n$-cover of the spectrum. Using the informations given in \cite{r1} about the number of band in periodic band spectra and in \cite{chin1} about the length  of the bands, we estimate $\varepsilon_n$ and give a bound for the number of band of this diameter . This yields to a bound from below of the minimal number of balls of diameter $\varepsilon_n$ one needs to cover the spectrum.
This bound also has  a direct  dynamical application and allows us to state a dynamical lower bound using the method in (\cite{tche1}). It is required for this lower bound to have the transfer matrix  norms polynomially bounded. This property is shown to be true for bounded density irrational number in \cite{ioc1}, hence more is not expected.
This limits dynamical implication of this lower bound to a set of irrational number of Lebesgue  measure $0$.

We will give precise statements of the model we study and our results in the next section. Section 3 will be devoted to the proof of our main result. We give a pathological example in the section 4 and a new lower bound for box counting dimension of the spectrum in section 5.

\section{Model and Statements}

 We limit our study to  the one dimensional  discrete Schrödinger operator $H_\beta$,
\begin{equation}\label{defH}
[H_\beta\psi](n) = \psi (n+1) + \psi (n-1) + V(n) \psi (n)
\end{equation}
acting on  $l^2(\Z)$, associated to a sturmian potential $V(n)$ given by
$$
V(n) = (\lfloor (n+1)\beta \rfloor - \lfloor n\beta \rfloor )V
$$
with  $\beta$ an irrational number in  $[0,1]$ and  $V$ a positive  constant.
We denote continued fraction expansion of $\beta$ by
$$
\beta = \frac{1}{a_1 + \frac{1}{a_2 + \dots}} = [0,a_1,a_2,\dots]
$$

The Fibonacci Hamiltonian, $H_\beta$ with $\beta=\frac{\sqrt{5}-1}{2} = [0,1,1,\dots]$ is the more simple example in Sturmian model because of its particular continued fraction development.

Since we are interested  by dynamical bounds, let us recall some quantities we want to bound:

We denote the time average outside probabilities
$$
P(N,T) = \sum_{|n| > N} a(n,T),
$$
with
$$
a(n,T) = \frac{2}{T} \int^\infty_0 e^{-2t/T} |\langle e^{-itH}\delta_1,
\delta_n \rangle |^2 dt.
$$
For all
$\alpha \in [0,+\infty ] $, see \cite{ger1}
$$
S^-(\alpha) = -\liminf_{T \to \infty } \frac{\log P(T^\alpha -2 ,T)
}{\log T}
$$
 and
$$
S^+(\alpha) = -\limsup_{T \to \infty } \frac{\log P(T^\alpha -2 ,T)
}{\log T}
$$

The following critical exponents are particular of interest:
$$
\alpha^{\pm}_l = \sup \{\alpha \geq 0: S^\pm (\alpha) =0\},
$$
$$
\alpha^{\pm}_u = \sup \{\alpha \geq 0: S^\pm (\alpha) < \infty \}.
$$

They verify $0\leq \alpha_l^\pm \leq \alpha_u^\pm$. In particular, if $\gamma
> \alpha_u^+$ then $P(T^\gamma,T)$ goes to $0$ fast.
$\alpha^{\pm}_l$ can be interpreted as the (lower and upper) rates of
propagation of the essential part of the wavepacket and $\alpha^{\pm}_u$
as  the rates of propagation of the fastest part of the wavepacket.

Moreover, we  always have for this kind of models $\alpha_u^+ \leq 1$. This upper bound, called ballistic, is
the fastest rate of spreading of the wavepacket.

Sturmian potentials (quasiperiodic structure) are the
buffer situation between random potentials (no structure in potential) that imply dynamical localization ($\alpha_u^\pm =0$) and periodic potentials that imply  ballistic spreading that is  $\alpha_u^\pm =1$.

More precisely, one has a non trivial strictly positive bound for almost all irrational numbers. In a sense we will make more precise latter, these irrational numbers are far enough from rational numbers. On the other hand, we show for  irrational number close enough to rational number, one has ballistic motion.

The first objective of this paper is to give a non ballistic upper bound for a large set of irrational numbers.

Recall the sequences associated to $\beta$:
$$
p_{-1} = 1, p_0 =0,
$$
$$
q_{-1} = 0,  q_0 = 1,
$$
$$
p_{k+1} = a_{k+1} p_k + p_{k-1}
$$
$$
q_{k+1} = a_{k+1} q_k + q_{k-1}
$$

 We can  now state our main result:

\begin{Th}\label{thprincipal}
Let $\beta$ be an irrational number and $H_\beta$ define as in  (\ref{defH})
with a sturmian potential associated to $\beta$. Assume that $V>20$.
 If $D
= \limsup_k \frac{\log q_k}{ k}$ is finite then
$$
\alpha_u^+ \leq \frac{2D}{\log \left( \frac{V-8}{3} \right) } .
$$

Moreover, for an irrational number  with  continued fraction expansion containing no $1$,
the dynamical upper bound becomes
$$
 \alpha_u^+ \leq \frac{D}{\log \left( \frac{V-8}{3} \right) }  .
$$
\end{Th}

\begin{req}
It is clear that taking $V$ large enough, one can obtain a non trivial bound that is smaller than $1$.
\end{req}

 It is well known that the set of irrational numbers satisfying $D$ finite is Lebesgue measure 1. For any algebraic number, that is with a periodic continued fraction development, one can easily compute $D$. Moreover, explicit valuation  for  $D$ is more general and we recall now a Khintchin interesting probability result:

\begin{Lem}(Khintchin)
For any $\beta \in [0,1]$, define the sequence
$$
q_{k+1} = a_{k+1}q_k + q_{k-1}, q_0=1, q_{-1}=0.
$$
For almost all $\beta$ with respect to Lebesgue measure,
 $$
D=\limsup_k\frac{\log q_k}{k}=D_K = \frac{\pi^2}{12 \log 2}.
$$
and
$$
M= \liminf_k (a_1 \dots a_k)^\frac{1}{k} = C_K =2.685...
$$
$C_K$ is  called the Khintchin constant.
\end{Lem}

\begin{Cor}
For Lebesgue almost every irrational numbers $\beta$, we have
$$
\alpha_u^+ \leq \frac{2D_K}{\log \left( \frac{V-8}{3} \right) }.
$$
\end{Cor}

\begin{proof}
It follows straightfully from previous theorem \ref{thprincipal} and Khintchin lemma.
\end{proof}

\begin{Cor}
For a precious number, that is $\omega$ = $[0,a,a,a,a,...]$, $a \neq 1$ the bound becomes
$$
 \alpha_u^+ \leq  \frac{\log(a + \omega)}{\log \left( \frac{V-8}{3} \right)}  .
$$
\end{Cor}

\begin{proof}
One can compute $q_k$ easily for such numbers.
\end{proof}

On the contrary, if $D$ is infinite, one can have ballistic motion at all large coupling: 

\begin{Th}
There exist an irrational number  $\omega$ with $D=+\infty$ such that
for any $V>20$ the dynamic  of $H_\omega$ is ballistic
\end{Th}

We also prove a new lower bound for the fractal dimension of the spectrum:

\begin{Th}
Denote  $C_k= \frac{3}{k} \sum_{j=1}^k \log (a_j +2)$. We have for
any irrational number  $\beta$ verifying
$C = \limsup C_k <+\infty$ and $V>20$:
\begin{equation}
dim_B^+ (\sigma) \geq \frac{1}{2} \frac{\log 2}{C + \log (V+5)}
\end{equation}
where $\sigma$ is the spectrum of $H_\beta$.
\end{Th}

\section{Proof of the theorem \ref{thprincipal}}

When one wants to bound all these dynamical quantities for specific models, it is useful to connect them to the qualitative behavior of the solutions of the difference equation 
\begin{equation}\label{diffeq}
\psi (n+1) + \psi (n-1) + V(n) \psi(n) = z \psi (n)
\end{equation}
with  $z \in \C$ and $\psi$ a non-zero vector.

One can reformulate this equation in terms of transfer matrix  .
$$
\binom{\psi (n+1)}{\psi(n)} = F (n,z) \binom{\psi(1)}{\psi(0)}
$$
with
$$
F(n,z) = \begin{cases}
           T(n,z)...T(1,z) & n \geq 1, \\
	   Id & n=0, \\
	   [T(n,z)]^{-1}....[T(0,z)]^{-1} & n \leq -1.
          \end{cases}	
$$
and
$$
T(m,z) = \begin{pmatrix}
           z - V(m) & -1 \\
	   1 & 0
	 \end{pmatrix}
$$

We denote
$$
M_k(z)  = F(q_k,z) = \begin{cases}
           T(q_k,z)...T(1,z) & n \geq 1, \\
	   Id & k=0, \\
	   [T(q_k,z)]^{-1}....[T(0,z)]^{-1} & n \leq -1.
          \end{cases}
$$
.

The following statement allows us to connect transfer matrix norms with dynamical exponents (see \cite{dam1} for details).

\begin{Th}
\label{thdam}
Let  $H_\beta$  be define as in  (\ref{defH}) and $K \geq 4$ is such that $\sigma ( H_\beta) \subseteq [-K+1,K-1]$.
Then, the outside probabilities can be bounded from above in terms of transfer matrix norms as follows:
$$
P_r(N,T) \lesssim \exp (-cN) + T^3 \int_{-K}^K \left( \max_{1\leq q_k \leq N} \left\|M_k(E + \frac{i}{T})
\right\|^2 \right)^{-1} dE,
$$
$$
P_l(N,T) \lesssim \exp (-cN) + T^3 \int_{-K}^K \left( \max_{-N\leq q_k \leq -1} \left\|M_k( E + \frac{i}{T})
\right\|^2 \right)^{-1} dE.
$$
the implicit constants depend only on  $K$ and $c$ is a universal positive constant.
\end{Th}

This theorem connects transfer matrix behavior with a dynamical upper bound in the following way. Choosing $N = N(T) = CT^\alpha$ such that the both integrals decay faster that any inverse power of $T$, implies that $P(N(T),T)$ goes to $0$ faster that any power inverse of $T$. By definition, of $\alpha_u^+$, it follows that $\alpha_u^+ \leq \alpha$. To exhibit such kind of condition, we have to prove the considered energy  is not in the spectrum, then the transfer matrix norm is shown to grow super exponentially.

We shall recall now a few properties of the transfer matrix and their traces.
The transfer matrix sequence verifies the evolution in $k$ (see e.g. \cite{s1},\cite{r1})
\begin{equation}
\label{dyn_mk}
M_{k+1}(z) = M_{k-1}(z) M_k(z)^{a_{k+1}}
\end{equation}

In order to bound from below the sequence of the norm of transfer matrix, it is enough to consider their traces. We recall now the following  result one can find in \cite{r1}.
 %DEFINITION t(k,p)

\begin{Prop}
Let $t_{k,p}$ be the trace of the matrix $M_{k-1}M_k^p$. The evolution along the $p$ index is given by
$$
t_{k,p+1} = t_{k+1, 0} t_{k, p} - t_{k, p-1},
$$
and consequently,
\begin{align}
t_{k,p+1} &=  S_{p} (t_{k+1,0})  t_{k,1} - S_{p-1} (t_{k+1,0})  t_{k,0}  \label{4r}\\
 \label{4r'}        &=  S_{p} (t_{k+1,0})  t_{k,0} - S_{p\pm 1} (t_{k+1,0})  t_{k,-1}
\end{align}
The evolution along the $k$ index is related to the p-evolution by
$$
t_{k+2,0} = t_{k,a_{k+1}},
$$
$$
t_{k+1,1} = t_{k,a_{k+1}+1},
$$
$$
t_{k+1,-1} = t_{k,a_{k+1}-1}.
$$

If one  denotes  by $x_k= t_{k+1,0}$ the trace of $M_k$ and $z_k = t_{k,1}$ the trace
of $M_{k-1}M_k$. This can be  reduce to the usual trace map relation  (\ref{4r})

$$
x_{k+1} = z_k S_{a_{k+1}-1}(x_k) - x_{k-1} S_{a_{k+1}-2}(x_k),
$$
$$
z_{k+1} = z_k S_{a_{k+1}}(x_k) - x_{k-1} S_{a_{k+1}-1}(x_k),
$$
with initial  conditions, $x_{-1}=2$, $x_0 = z$ et $z_0 = z -V$.
\end{Prop}

\begin{req}
This two sequences are dependent on $z$ but we will omit it in order to simplify notations.

Here, $S_l$ denotes the $l^{th}$ Tchebychev polynomial of the second kind:
$$
S_{-1}(x) =0,
$$
$$
S_{0}(x) =1,
$$
$$
S_{l+1}(x) = x S_l(x) - S_{l-1}(x), \hspace{5mm} \forall l \geq 0.
$$
%We have the following property:
%$$
%S_l (2\cos (\theta)) = \frac{\sin((l+1)\theta)}{\sin (\theta)}, \forall
%l \geq -1, \forall \theta \in \C.
%$$

\end{req}

%spectre periodic suto%
The sequence $\{x_k (z)\}_k $ can have two different behaviors depending on $z$. If and only if $z$ lies in the spectrum of $H_\beta$ then this sequence is bounded. A criterium has first been stated by S\"{u}t\H{o} in $\cite{s1}$ for Fibonacci Hamiltonian and extended by Bellissard et Al. in \cite{bel} for other irrational numbers. The attendance of $\delta$ is purely technical and doesn't change the proof.

\begin{Lem}\label{expl}
A necessary and sufficient condition that $\{x_k (z)\}_k $ be unbounded is that
$$
x_{N-1} (z) \leq 2+ \delta , x_{N} (z) > 2+ \delta, z_{N} (z) > 2+ \delta
$$
for some $N \geq 0$. This $N$ is unique.
Denote
$$
G_k = G_{k-1} + a_k G_{k-2}, G_0=1, G_{-1}=1.
$$
We have
$$
|x_{k+1}| \geq |z_k| \geq e^{c G_{k-N}} +1 \hspace{5mm} \forall k > N.
$$
with $c= \log (1+\delta) >0$  constant.
\end{Lem}

This criterium motivates the following definition:

Denote $\sigma_{k,p} = \left\{ E \in \mathbb{R}, |t_{k,p}(E)| \leq 2  \right\} $.

Denote also $\beta_n = \frac{p_n}{q_n}$ , the rational approximation of $\beta$. It is well known the set $\sigma_{k,0}$ coincide with the spectrum of the operator $H_{\beta_n}$, where $\beta_n$ replace $\beta$ in the definition of $H_\beta$.  The sequence of operator $\{H_{\beta_n}\}$ is called the periodic approximants of $H_\beta$ and converges strongly to $H_\beta$. It is well known spectrum of $H_\beta$ is a Cantor set that can be approximate by the band spectra of the periodic approximants. The following proposition recalls precisely this statement (\cite{s1}, \cite{bel}), \cite{tes}):

\begin{Prop}
The sequence of spectra of periodic approximants of $H_\beta$ is such that

(i) the set $\sigma_{k,p}$ is made of $p q_k + q_{k-1}$ distinct
intervals,

(ii) $\sigma \subset \sigma_{k+1,0} \cup \sigma_{k,0}$ and \hspace{2mm}$\sigma_{k,p+1} \subset \sigma_{k+1,0} \cup \left( \sigma_{k+1,0}^c  \cap \sigma_{k,p}\right), \forall k \in \mathbb{N},$

(iii) $\sigma_{k+1,0}  \cap \sigma_{k,p}\cap \sigma_{k,p-1}= \emptyset, \forall V > 4 \hspace{2mm} and \hspace{2mm} \forall k \in \mathbb{N}, p \geq 0$.
\end{Prop}

We recall now important result when one gets interest on periodic approximants spectra. It allows to know the way the intervals of  $\sigma_{k,p}$ are included in $\sigma_{k-1,p}$. It requires some definitions:

%dynamik bande raymond%
\begin{definition} For a given $k$, we call

type I gap : a band of $ \sigma_{k,1}$ included in a band of $ \sigma_{k,0}$ and therefore in a gap of $ \sigma_{k+1,0}$,

type II band : a band of $ \sigma_{k+1,0}$ included in a band of $ \sigma_{k,-1}$ and in a gap of $ \sigma_{k,0},$

type III band : a band of $ \sigma_{k+1,0}$ included in a band of $ \sigma_{k,0}$ and in a gap of $ \sigma_{k,1}$.
\end{definition}

It is proved in \cite{r1} these definitions exhaust all the possible configuration with the  following lemma

\begin{Lem}(Raymond) \label{raymond}
At a given level $k$,

(i) a type I gap contains an unique type  II band of $\sigma_{k+2,0}$.

(ii) a type II band contains $(a_{k+1}+1)$ bands of type I of $\sigma_{k+1,1}$. They are alternated with $(a_{k+1})$ type III bands of $\sigma_{k+2,0}$

(iii) a type III band contains $(a_{k+1})$ bands of type I of $\sigma_{k+1,1}$. They are alternated with $(a_{k+1}-1)$ type III bands of $\sigma_{k+2,0}$
\end{Lem}

As stated above, $H_{\beta_n}$ spectrum is made by a growing number of intervals of decreasing length as $n$ is increasing.
We recall now a result obtain in \cite{chin1} which allows to control the length of the bands of $\sigma_{k,p}$ at any level $k$. We need again some notations to resume it:

Let  $\mathcal{A} = \{I,II,III\}$ be an alphabet. For each band $B$ of spectrum at level $k$, correspond an unique word $i_0i_1\dots i_k \in \mathcal{A}^{n+1}$ such that $B$ is  a band of type $i_k$ included in a band of type
$i_{k-1}$ at level $k-1$,...,included in a band of type $i_0$ at level 0.
This word will be called the $index$ of $B$. More than one band can have the same index.
Let $T_n = (t_{i,j}(n))_{3*3}$ be a sequence of matrix and $\tau = i_0i_1\dots i_k$ an index, we define:
$$
L_\tau (T) = t_{i_0,i_1}(1)t_{i_1,i_2}(2)\dots t_{i_{k-1},i_k}(k).
$$

We can now recall the result in \cite{chin1}:

\begin{Th}\label{thchinois}(Liu, Wen)
 If $\beta = [a_1,a_2...]$ be an irrational number in  $[0,1]$ and $H_\beta$ define as above with $V> 20$ then any band   $B$ of index $\tau$ verify,
$$
4 L_\tau (Q) \leq |B| \leq 4 L_\tau (P)
$$

 where $P=(P_n)_{n>0}$
$$
P_n = \begin{pmatrix}
0 & c_1^{a_n-1} & 0 \\ c_1/a_n & 0 & c_1/a_n \\c_1/a_n & 0 & c_1/a_n
\end{pmatrix}
$$
with $c_1= \frac{3}{V-8}$
and $Q=(Q_n)_{n>0}$
$$
Q_n =\begin{pmatrix}
0 & c_2^{a_n -1} & 0 \\ c_2 (a_n +2)^{-3} & 0 & c_2 (a_n +2)^{-3} \\ c_2 (a_n
+2)^{-3} &0 & c_2 (a_n +2)^{-3} \\
\end{pmatrix}
$$
with $c_2= \frac{1}{V+5}$.
\end{Th}

%taille des bandes chinois%

%spectre complexe

By now, we define the periodic approximants spectrum not only in $\mathbb{R}$ but in $\mathbb{C}$.
$$
\sigma_{k,0}^\delta = \{z \in \mathbb{C} : |x_k(z)| \leq 2 + \delta\}
$$

All the proposition above keep true replacing $\sigma_{k,p}$ by $\sigma_{k,p}^\delta$ for some small enough fixed $\delta$. A condition on $V$  should be added to keep the  invariant formula, $V > \lambda(\delta) = [12(1+\delta)^2+8(1+\delta)^3+4]^{1/2}$.

The following proposition states, due to classical Koebe distortion theorem, the height of this set  is almost the same that its length.

\begin{Prop}
If $k \geq 3$, $\delta >0$ and $V>20$ then there exists constants
$c_\delta$,$d_\delta>0$ such that
$$
\bigcup_{j=1}^{q_{k-1}} B ( x_k^{(j)},r_k) \subseteq \sigma_{k,0}^\delta
\subseteq \bigcup_{j=1}^{q_{k-1}} B ( x_k^{(j)},R_k)
$$
where $\{x_k^{(j)}\}_{1 \leq j \leq q_{k-1}}$ are the zeros of $x_k$,
$r_k = c_\delta \inf_{\tau}  L_\tau (Q)$ and $ R_k = d_\delta \sup_{\tau}  L_\tau(P)$.
\end{Prop}

\begin{proof}
The proof follows the same steps that in \cite{dam1}. Let $C_j$ be a connected component of  $\sigma_{k,0}^{2\delta}$. With  $V> \max \{20, \lambda(2\delta)\}$, it is easy to see that $C_j$ contains
exactly one of a $q_{k-1}$ zeros of  $\sigma_{k,0}^\delta$, $x_k^{(j)}$.
And it is clear that  $C_j$ contains one connected component of
$\sigma_{k,0}^\delta$, denoted by $\tilde{C_j}$.  It suffice
to show that
\begin{equation}
B ( x_k^{(j)},r_k) \subseteq \tilde{C_j}
\subseteq  B ( x_k^{(j)},R_k).
\end{equation}
to obtain the result.

As $x_k$ is a proper function (as a polynom of $z$)  and  $C_j$ contains an unique
zero, its degree is 1.
$$
x_k : \text{int} (C_j) \to B(0,2 + 2\delta)
$$
is univalent and so
$$
x_k^{-1} : B(0,2+2\delta) \to \text{int} (C_j)
$$
is well define and univalent too.
Consequently, the function
$$
F : B(0,1) \to \C, F(z) = \frac{x_k^{-1}((2+2\delta)z) - x_k^{(j)}}{(2+2\delta)(x_k^{-1})'(0)}
$$
is univalent on $B(0,1)$. clearly, we have $F(0)=0$ and $F'(0)
=1$.

Applying  Koebe distortion theorem , we get
$$
\frac{|z|}{(1+|z|)^2} \leq |F(z)| \leq \frac{|z|}{(1-|z|)^2}, |z| \leq 1.
$$

Evaluating this for $|z| = \frac{2+\delta}{2+2\delta}$, one has
$$
\frac{(2+\delta)(2+2\delta)}{(4+3\delta)^2} \leq F(z) \leq \frac{(2+\delta)(2+2\delta)}{\delta^2}
$$
By definition of $F$ this implies
$$
|x_k^{-1} ((2+2\delta)z) - x_k^{(j)}| \leq \frac{(2+\delta)(2+2\delta)}{\delta^2}
|(x_k^{-1})'(0)|
$$
$$
|x_k^{-1} ((2+2\delta)z) - x_k^{(j)}| \geq \frac{(2+\delta)(2+2\delta)}{(4+3\delta)^2}
|(x_k^{-1})'(0)|
$$
And then for  $|z|=2+\delta$,
$$
|x_k^{-1} (z) - x_k^{(j)}| \leq \frac{(2+\delta)(2+2\delta)}{\delta^2}
|(x_k^{-1})'(0)|
$$
$$
|x_k^{-1} (z) - x_k^{(j)}| \geq \frac{(2+\delta)(2+2\delta)}{(4+3\delta)^2}
|(x_k^{-1})'(0)|
$$
It suffices with $|(x_k^{-1})'(0)| = |x_k'(x_k^{(j)})|$ to remark that
$$
r_k \leq |(x_k^{-1})'(0)| \leq R_k
$$
and with $|z|=2+\delta$, $x_k^{-1} (z)$ runs through the entire boundary of $\tilde{C_j}$ to conclude.

\end{proof}

We stated the  following technical  lemma to replace  $\{G_k\}_k$ by $\{q_k\}_k$
\begin{Lem}\label{lemma4}
If $D
= \limsup_k \frac{\log q_k}{ k}$ is finite  then  there exists a constant
$d>1$ such that
$$
G_k^d > q_k
$$
for all level  $k$.
\end{Lem}

\begin{proof}
By assumption, there exists   a constant $c$ 
such that
$$
q_k \leq c^k
$$
for all $k$.
It is clear by the definitions of the two sequences that
$$
G_k \leq q_k
$$
and that $G_k$ is bigger that the fibonacci sequence. It suffices to choose $d$ greater than  $\frac{\log c}{\log
\omega}$ where $\omega$ is the golden mean.
\end{proof}

%demonstration finale%

\begin{proof}[Proof of theorem \ref{thprincipal}]

We have now all the  required tools to finish the proof of the theorem \ref{thprincipal}.

As $x_k^{(j)}$ are real, we have
$$
\sigma_{k,0}^\delta \subseteq \{ z \in \C \text{:} |Im   z | < R_k\} \subseteq \{z \in \C \text{:} |Im
 z | <d
q_k^{-\gamma(V)}\}.
$$
for a suitable $\gamma (V)$ .
This implies %la proposition \ref{prop_inv}
\begin{equation}
\label{eq_expl}
\sigma_{k,0}^\delta \cup \sigma_{k,1}^\delta \subseteq \{z \in \C \text{:} |Im  z | < d
q_k^{-\gamma(V)}\}.
\end{equation}

Let us precise how to choose  $\gamma(V)$.

We need to major $R_k$.
 $R_k$ is  the  supremum of products of $k$ elements
of matrix $P_n$. All the coefficients in $P_n$ are maximal for $a_n=1$. The worst case possible  happens when a band has a index history
type I containing a band of type II, in that case the coefficient
 could be trivial equal to 1 (if $a_n=1$). But because of combinatoric behavior of bands described by the proposition \ref{raymond}, this situation
can't occur more than half of the time.
Consequently this implies
$$
R_k \leq c_1^{k/2}
$$

We should have
$
R_k < d q_k^{-\gamma(V)}
$
so a suitable $\gamma$ can be chosen  by taking:
$$
\gamma (V) \leq \limsup_k  - \frac{k\log c_1}{2  \log q_k }.
$$

For $\varepsilon = Im z >0$, we get a lower  uniform lower bound for $|x_n(E+i\varepsilon)|$ with $E \in [-K,K] \subset \R$. For a fixed
$\varepsilon>0$, we choose $k$ such that $ d
q_k^{-\gamma(V)} < \varepsilon$. With (\ref{eq_expl}), this shows $|x_k (E +
i\varepsilon)| > 2+\delta$ and $|z_k (E +i\varepsilon)| > 2+\delta$. As $|x_{-1}(E +
i\varepsilon)| =2 \leq 2+\delta$ we are in the situation of the proposition \ref{expl}
and we have the bound
$$
|x_j| \geq e^{\log(1+\delta) G_{j-k}} +1 \hspace{5mm} \forall j > k.
$$
And applying Lemma \ref{lemma4}, we obtain
\begin{equation}
|x_j| \geq e^{\log(1+\delta) q_{j-k}^d} +1 \hspace{5mm} \forall j > k.
\end{equation}

All this motivates the following definitions:
\begin{definition}
 For $\delta >0, T>1$, denote by $k(T)$  the unique
integer with
$$
\frac{q_{k(T)-1}^{\gamma(V)}}{d_\delta} \leq T \leq \frac{q_{k(T)}^{\gamma(V)}}{d_\delta}
$$
and let
$$
N(T) = q_{k(T) + \lfloor \sqrt{k(T)} \rfloor}.
$$
\end{definition}

It is then easy to see for $T$ large enough and 
for every  $\nu > 0$, that we have a constant $C_\nu >0 $ such that
\begin{equation}{\label{estiN(T)}}
N(T) \lesssim C_\nu T^\frac{1}{\gamma (V)} T^\nu.
\end{equation}
Applying theorem \ref{thdam} and above estimate, we get
\begin{align*}
P_d (N(T),T) &\lesssim exp(-cN(T)) + T^3 \int_{-K}^K \left( \max_{1\leq q_n \leq N(T)}
\left\| M_n( E + \frac{i}{T})
\right\|^2 \right)^{-1} dE,) \\
	     &\lesssim exp(-cN(T)) + T^3 e^{-2\log(1+\delta) G_{\lfloor
\sqrt{k(T)} \rfloor}} \\
            &\lesssim exp(-cN(T)) + T^3 e^{-2\log(1+\delta) q_{\lfloor \sqrt{k(T)}
\rfloor}^{d}}.
\end{align*}

From this bound, it is clear that  $P_d (N(T),T)$ goes to zero faster than any inverse power. One gets the same bound for  $P_g (N(T),T)$ because of the symmetry of the potential. Finally, one can conclude  that 
$$
\alpha_u^+ \leq \alpha
$$
with
$$
\alpha = \frac{1}{\gamma(V)} + \nu
$$
and  $\nu$ arbitrary small.

For the second part of the theorem, notice the constant $2$ appears from the choice of $\gamma(V)$ considering the worst coefficient in matrix $P_n$. But assuming there no $1$ in continued fraction development, one gets
$$
R_k \leq c_1^{k}
$$
and
$$
\gamma (V) \leq \limsup_k  - \frac{k\log c_1}{\log q_k }.
$$

\end{proof}

%fin de la demonstration principal

\section{A pathological counter-example}
The statements above holds if $D< + \infty$. In the case
$D=+\infty$, we exhibit in the next
statement a counter example. It is still an open question if $D=+\infty$ implies ballistic motion.

\begin{Th}\label{thexemple}
There exists an irrational number  $\omega$ with $D=+\infty$ such that
for any $V>20$.
$$
\alpha_u^+ = 1.
$$
\end{Th}

The proof,  made by induction, follows the lines of pathological example in \cite{Last}. The main idea is that, choosing an irrational number  close to rational numbers (with large values for the sequence $\{a_k\}_k$), potentials of $H_\beta$ and  $H_{\beta_n}$ coincide on large scale of time. Large enough to say that $H_\beta$ and  $H_{\beta_n}$ have the same dynamical behavior.
It is well known that periodic operator $H_{\beta_n}$ has ballistic motion.

We make now these ideas more precise and first prove the following lemma:

\begin{Lem}
Sturmian potentials of operator $H_\beta$ and $H_{\beta_n}$
have  the same first  $q_{n+1}$ values.
\end{Lem}

\begin{proof}
To prove this, we  recall the iterative construction of sturmian word that coincide with our potential. For details and proof, see e.g. \cite{Lothaire}.
Denote $W_0=0$ et $W_1=0^{a_1 -1} V$ and define the sequence of sturmian  words by
$$
W_{k+1} = W_k^{a_{k+1}}W_{k-1},\hspace{5mm} k \geq 1.
$$
Each word $W_k$ has length  $q_k$.

As $H_\beta$ and $H_{\beta_n}$ have the same first $n$ terms of continued fraction expansion, words $W_0, W_1,...W_n$ are the same for $H_\beta$ and
$H_{\beta_n}$.

For $H_{\beta_n}$, the limit word $W_\infty$ is periodic with period $q_n$ and repeat
endless the word $W_n$.
As $W_n = W_{n-1}^{a_n}W_{n-2}$, one has
$$
W_n^\infty = W_n^{a_{n+1}}  W_{n-1}^{a_n}W_{n-2} W_n^\infty.
$$
This shows that the potential  $H_{\omega_n}$ begins with the word
 $W_n^{a_{n+1}}  W_{n-1}$ which is the word  $W_{n+1}$ for $H_\omega$.
  As $W_{n+1}$ is $q_{n+1}$ long, this ends the proof.
\end{proof}

We need another lemma, one can find in \cite{Last}. It states that two operators have close dynamic (on some scale of time $T$) if their potentials are close enough. We make this idea more precise by recalling this lemma:

\begin{Lem}\label{lemlast}
Let $H_1 = \Delta + V_1  $ and $H_2 = \Delta +V_2$ acting on $l^2(\Z),$ and such that  $|V_1(k)|, |V_2(k)| < C$
for all  $k \in  \Z$ and some constant  $C$. Let $T > 0$ and $\varepsilon > 0$ be fixed constant then if it exists $L(T,\varepsilon), \delta > 0 $ such that $|V_1(k) - V_2(k)| < \delta$ for all $|k| < L$, then
$$
| \langle |X|^2_{H_1} \rangle_T -  \langle |X|^2_{H_2} \rangle_T | < \varepsilon.
$$
\end{Lem}

We get back to the construction.

\begin{proof}[Proof of theorem \ref{thexemple}]
As $H_{\omega_n}$ is a periodic potential  operator, one has
$$
\langle |X|^2_{H_{\omega_n}} \rangle_T > C_n T^2.
$$
choose  $T_n$ big enough such that
$$
C_n > \frac{1}{\log T_n}.
$$

One can then choose $a_{n+1}$ such that  $L(T_n,1) \leq q_{n+1}$.
Then Lemma \ref{lemlast} implies 
\begin{equation}
\langle |X|^2_{H_\omega} \rangle_{T_n} > \frac{T_n^2}{\log T_n} -1.
\end{equation}
Inductively, we have a sequence $T_n$ going to infinity and an irrational number $\omega$ with
$$
\langle |X|^2 \rangle_{T_n} > \frac{ T_n^2}{\log T_n}-1 > T_n^{2-\varepsilon}, \forall \varepsilon >0
$$
which yields to
$$
\alpha_u^+  \geq \beta^-_{\delta_1} (2) > 1-\varepsilon, \forall \varepsilon >0.
$$
\end{proof}

\section{Lower bound for the box counting dimension of the spectrum}

We give now a lower  bound of the fractal box counting  dimension of the spectrum of operator $H_\beta$.
We recall now the defintion. If one denotes $N(\varepsilon)$ the number of balls of diameter at most $\varepsilon$ one need to cover $\sigma$, then the upper box counting dimension is defined by
$$
\dim_B^+ = \limsup_{\varepsilon \to 0}\frac{\log N(\varepsilon)}{\log \varepsilon}
$$

 The spectrum is approached by the band spectrum of periodic  $H_{\beta_n}$. Moreover, in \cite{r1,chin1}, we have precise informations of the number of bands and their length. It allows  to give a lower bound of minimal number of set of some decreasing scale needed to cover the spectrum and then to give a lower bound of box dimension of the limit set.  The first idea to cover the spectrum can be  to take into account all the bands and take  as a scale the smaller length, but this is a bad idea because this minimal length  decreases faster than the number of intervals grows. The second idea can be to count the number of bands that have the maximal length, in terms of inverse power of $V$. This yields to a better lower bound for the box dimension of the spectrum for almost every irrational numbers.

Fixing the irrational number, one can improve this method, by counting precisely the number of band that have a particular length. It has been made for Fibonacci number in \cite{tche2} where the full fractal spectrum has been investigated. The length of a band is depending of its history, in that case, the times it gets by a type I band. Hence, one obtains this way all the contribution at any scale to the box dimension. It is shown their result is optimal with $V$ increasing and one has for $\beta = [0,1,1, \dots]$
$$
\dim_B (\sigma (H_\beta))\approx \frac{\log (1 +\sqrt{2})}{\log V}.
$$

An other example, simpler than golden  mean is silver ratio. Fix $\beta=[0,2,2,\dots]$, then all the bands have the same length up to a constant independent of $V$. Namely, all bands at level $k$ have length $c_k V^{-k}$, where $c_k$ is a constant depending of  history of the band but not of $V$.

This implies that one has:
$$
\dim_B (\sigma (H_\beta)) \geq -\liminf_k \frac{\log q_k}{\log c_k V^{-k}} \approx \frac{\log (1 +\sqrt{2})}{\log V}.
$$

It is easy to show with same kind of argument the other side inequality and hence we obtain the same estimation for this case
$$
\dim_B (\sigma (H_\beta))\approx \frac{\log (1 +\sqrt{2})}{\log V}.
$$
It is quite astonishing that both golden mean and silver ratio yield  the same fractal dimension estimate.

Going back to the general case, we will apply the same method used for silver mean, that is count the number of band at level $k$ that have length equal to $c_k V^{-k}$. We obtain:

\begin{Th}\label{thdimboite}
Denote  $C_k= \frac{3}{k} \sum_{j=1}^k \log (a_j +2)$. We have for
any irrational number  $\beta$ verifying
$C = \limsup C_k <+\infty$ and $V>20$:
\begin{equation}\label{borneinf}
dim_B^+ (\sigma) \geq \frac{1}{2} \frac{\log 2}{C + \log (V+5)}
\end{equation}
where $\sigma$ is the spectrum of $H_\beta$.
\end{Th}

\begin{req}
As in theorem \ref{thprincipal}, $C$ finite is a full condition equivalent to $D$ finite.
\end{req}

The following lemma give precise statement of the counting idea.
\begin{Lem}
Denote $n_{k,I}, n_{k,II}$ and  $n_{k,III}$ the number  of bands of type respectively I, II et  III in respectively $\sigma_{k,1}, \sigma_{k+1,0}, \sigma_{k+1,0}$ and with a length  greater than
$\varepsilon_k =4 \Pi_{j=1}^k (V+5)^{-1} (a_j +2)^{-3}$.

For all  $k$, we have the following induction relation:
$$
n_{k+1,I} = (a_{k+1} +1)n_{k,II} + a_{k+1} n_{k,III},
$$
$$
n_{k+1,II} = 1_{\{a_{k+1} \leq 2\}} n_{k,I},
$$
$$
n_{k+1,III} = a_{k+1} n_{k,II} + (a_{k+1}-1) n_{k,III}.
$$
With initials conditions $n_{0,I}=1, n_{0,II}=0,  n_{0,III}=1$.

Moreover  this three sequences verify the following properties:

$$
n_{k,II} \neq 0 \bigvee  n_{k,III} \neq 0
$$
$$
n_{k,I} \neq 0
$$
$$
n_{k,I} > n_{k,III}
$$
and
$$
n_{k,II} + n_{k,III} > 2^{\lfloor \frac{k}{2} \rfloor}
$$

\end{Lem}

%% PREUVE DU LEMME SUR LE NOMBRE DE BANDE DE BONNES TAILLES
\begin{proof}
The induction relation is obvious with (\ref{thchinois}).

The two first properties are made by induction.
Initial conditions give level $0$.
 Assume it is true at level $n$, then  as $a_{k+1} > 0$,  $n_{k,II} \neq 0 \bigvee  n_{k,III} \neq
0$, implies   $n_{k+1,I} \neq 0$.
For the second part, if  $a_{k+1} \leq 2$ then $n_{k+1,II} \neq 0$, else $a_{k+1}
> 2$ implies $n_{k+1,III} \neq 0$.

To prove
$$
n_{k,I} > n_{k,III}
$$
it suffices to see that
$$
n_{k,I} =  n_{k,III} + n_{k-1,II} + n_{k-1,III}
$$
.

For the last property, it suffices to show that
$$
n_{k,II} + n_{k,III} \geq 2 (n_{k-2,II} + n_{k-2,III})
$$

Using induction relation, we get
$$
n_{k,II} = [(a_{k-1} +1) n_{k-2,II} + a_{k-1}n_{k-2,III}] 1_{\{a_k \leq 2\}}
$$
$$
n_{k,III} = (a_k-1) (a_{k-1} n_{k-2,II} + (a_{k-1}-1)n_{k-2,III}) + a_k
n_{k-2,I} 1_{\{a_{k-1} \leq 2\}}
$$
We distinguish 4 cases depending on the values of $a_k$ and
$a_{k-1}$.

$\bullet$
If $a_k > 2$ and $a_{k-1} >2$, then we simply get
\begin{align*}
n_{k,II} + n_{k,III} &= (a_k-1) (a_{k-1} n_{k-2,II} +
(a_{k-1}-1)n_{k-2,III}) \\
&\geq (a_k-1)(a_{k-1}-1)(n_{k-2,II} + n_{k-2,III}) \\
&\geq 4 (n_{k-2,II} + n_{k-2,III}).
\end{align*}

$\bullet$
If $a_k \leq 2$ and $a_{k-1} > 2$, then one has
\begin{align*}
n_{k,II} + n_{k,III} &= (a_k-1) (a_{k-1} n_{k-2,II} + (a_{k-1}-1)n_{k-2,III}) \\
 & \hspace{5mm}+ (a_{k-1} +1) n_{k-2,II} + a_{k-1}n_{k-2,III} \\
&\geq a_k a_{k-1} (n_{k-2,II} + n_{k-2,III}) \\
&\geq 3 (n_{k-2,II} + n_{k-2,III})
\end{align*}

$\bullet$
If $a_k > 2$ and $a_{k-1} \leq 2$, then one has
\begin{align*}
n_{k,II} + n_{k,III} &= (a_k-1) (a_{k-1} n_{k-2,II} + (a_{k-1}-1)n_{k-2,III}) + a_k
n_{k-2,I} \\
&\geq (a_k-1) (a_{k-1} n_{k-2,II} + (a_{k-1}-1)n_{k-2,III}) + a_k
n_{k-2,III} \\
&\geq (a_k-1)a_{k-1} (n_{k-2,II} + n_{k-2,III}) \\
&\geq 2 (n_{k-2,II} + n_{k-2,III})
\end{align*}

$\bullet$
If $a_k \leq 2$ and $a_{k-1} \leq 2$, then one gets
\begin{multline*}
n_{k,II} + n_{k,III} = (a_k-1) (a_{k-1} n_{k-2,II} + (a_{k-1}-1)n_{k-2,III}) +\\ a_k
n_{k-2,I} + (a_{k-1} +1) n_{k-2,II} + a_{k-1}n_{k-2,III}
\end{multline*}
And one obtains
\begin{align*}
n_{k,II} + n_{k,III} &\geq ((a_k-1)a_{k-1} + a_{k-1} +1) n_{k-2,II} \\
& \hspace{5mm}+ 
((a_{k}-1)(a_{k-1}-1) + (a_{k-1} + a_{k}) n_{k-2,III} \\
&\geq (a_k a_{k-1} +1)n_{k-2,II} + (a_{k-1} + a_{k})n_{k-2,III} \\
&\geq 2 (n_{k-2,II} + n_{k-2,III}).
\end{align*}
\end{proof}

\begin{proof}[Proof of theorem \ref{thdimboite}]
With previous lemma, we find a bound for $n_{k,II} +
n_{k,III}$, that is the number of bands of length at least $\varepsilon_k$. To make sure we have a disjoint cover we consider only half of the bands. Each band is then separeted by another band we does not count.
Then by definition of box dimension, we have
$$
\dim_B^+ (\sigma) \geq \liminf_k \frac{\log 1/2(n_{k,II} +
n_{k,III})}{-\log \varepsilon_k}.
$$
and the stated result.

\end{proof}

\begin{req}
The former  bound for box dimension   provided in \cite{chin1} was
$$
dim_B^+ (\sigma) \geq dim_{H} (\sigma) \geq \max \left\{ \frac{\log 2 }{10 \log 2 -
3 \log t_2}, \frac{\log M - \log 3}{\log M - \log t_2/3} \right\}.
$$
where $M = \liminf_{k \to \infty} (a_1a_2\dots a_k)^\frac{1}{k}$ and
$t_2=\frac{1}{4(V+8)}$.

For almost all irrational numbers, that is  with $M$ equal to the Khintchin constant
(2.685...), our bound is better than  above and for any $V>20$. On the other
hand, for all fixed $V$, one has no improvement with some specific numbers. Fixing $\beta=[0,c,c,\dots] $, the
bound above goes
to 1 and (\ref{borneinf}) to 0 as $c$ goes to infinity. 
\end{req}

A lower bound for box dimension can be relevant to obtain a lower bound
for dynamic exponent $\alpha_u$.

\begin{definition}
A irrational number is said to be a bounded density irrational number if it fulfills the following condition
$$
\limsup_n \frac{1}{n} \sum_{i=1}^n a_i < + \infty. 
$$
\end{definition}

\begin{Th}
For any bounded density irrational number, we have
$$
\alpha_u^- \geq\frac{1}{2} \frac{\log 2}{C + \log (V+5)}
$$
with $C=\limsup\frac{3}{k} \sum_{j=1}^k \log (a_j +2)$.
\end{Th}

\begin{proof}
 It is shown in \cite{tche1} and
\cite{tche2} that if the norms of the transfer matrix are polynomially
bounded on the spectrum then
we have $\alpha_u^- \geq \dim_B^+(\sigma)$. This property on  the
norm of the transfer matrix is shown for irrational  with
bounded density in \cite{ioc1}.
\end{proof}

\noindent\textit{Acknowledgments.} It is a pleasure to thank Dominique Vieugué for useful
conversations about number theory.

\end{document}